\documentclass[a4paper,UKenglish]{lipics-v2018}
\nolinenumbers
 
\usepackage{microtype}

\newtheorem{fact}[theorem]{Fact}

\usepackage{latexsym} 
\usepackage{amsmath,amssymb}
\usepackage{comment}
\usepackage{url}
\usepackage{subcaption}
\usepackage{makeidx}
\usepackage{algorithm}
\usepackage{algorithmicx}
\usepackage[noend]{algpseudocode}
\usepackage{dcolumn}
\usepackage{multirow}

\newcommand{\floor}[1]{\left \lfloor #1 \right \rfloor}
\newcommand{\ceil}[1]{\left \lceil #1 \right \rceil}

\newcommand{\alphabet}{\mathcal{A}}

\newcommand{\rlesize}[1]{|{#1}|_{\mathsf{rle}}}

\newcommand{\emptystr}{\varepsilon} 

\newcommand{\slp}{\mathcal{S}} 
\newcommand{\repair}{\mathsf{RePair}} 
\newcommand{\vars}{\mathcal{V}}
\newcommand{\rules}{\mathcal{D}}
\newcommand{\hvars}[1]{\mathcal{V}_{#1}}
\newcommand{\hrules}[1]{\mathcal{D}_{#1}}

\newcommand{\dtree}{\mathcal{T}} 
\newcommand{\gtext}[1]{\mathit{T}_{#1}} 
\newcommand{\letters}[1]{\Sigma_{#1}} 

\newcommand{\lchar}[1]{\acute{#1}}
\newcommand{\rchar}[1]{\grave{#1}}

\newcommand{\isSB}[1]{\mathsf{isSB}({#1})} 
\newcommand{\lmb}[1]{\lambda({#1})} 
\newcommand{\rmb}[1]{\rho({#1})} 

\newcommand{\NULL}{\mathsf{NULL}}
\newcommand{\pol}{\mathsf{PopOutLet}} 
\newcommand{\pil}{\mathsf{PopInLet}} 

\newcommand{\freq}{\mathsf{freq}}
\newcommand{\vocc}{\mathsf{vocc}} 

\newcommand{\rg}[1]{\mathcal{S}_{#1}} 

\newcommand{\val}[1]{\mathit{val}_{#1}} 

\title{RePair in Compressed Space and Time
}

\author{Kensuke~Sakai}{Kyushu Institute of Technology, {680-4 Kawazu, Iizuka, Fukuoka 820-8502, Japan}}{k\_sakai@donald.ai.kyutech.ac.jp}{}{}
\author{Tatsuya~Ohno}{Kyushu Institute of Technology, {680-4 Kawazu, Iizuka, Fukuoka 820-8502, Japan}}{t\_ohno@donald.ai.kyutech.ac.jp}{}{}
\author{Keisuke~Goto}{Fujitsu Laboratories Ltd., 4-1-1 Kamikodanaka, Nakahara, Kawasaki, Kanagawa 211-8588, Japan}{goto.keisuke@jp.fujitsu.com}{https://orcid.org/0000-0001-6964-6182}{}
\author{Yoshimasa~Takabatake}{Kyushu Institute of Technology, {680-4 Kawazu, Iizuka, Fukuoka 820-8502, Japan}}{takabatake@ai.kyutech.ac.jp}{}{}
\author{Tomohiro~I}{Kyushu Institute of Technology, {680-4 Kawazu, Iizuka, Fukuoka 820-8502, Japan}}{tomohiro@ai.kyutech.ac.jp}{https://orcid.org/0000-0001-9106-6192}{Supported by JSPS KAKENHI Grant Number JP16K16009.}
\author{Hiroshi~Sakamoto}{Kyushu Institute of Technology, {680-4 Kawazu, Iizuka, Fukuoka 820-8502, Japan}}{hiroshi@ai.kyutech.ac.jp}{}{}

\authorrunning{
  K. Sakai, T. Ohno, K. Goto, Y. Takabatake, T. I, and H. Sakamoto
}

\Copyright{
Kensuke Sakai, Tatsuya Ohno, Keisuke Goto, Yoshimasa Takabatake, Tomohiro I, and Hiroshi Sakamoto
}

\subjclass{
  Data structures design and analysis $\rightarrow$ Data compression
}


\category{}

\relatedversion{}

\supplement{}

\funding{}

\acknowledgements{}

\EventEditors{John Q. Open and Joan R. Access}
\EventNoEds{2}
\EventLongTitle{42nd Conference on Very Important Topics (CVIT 2016)}
\EventShortTitle{CVIT 2016}
\EventAcronym{CVIT}
\EventYear{2016}
\EventDate{December 24--27, 2016}
\EventLocation{Little Whinging, United Kingdom}
\EventLogo{}
\SeriesVolume{42}
\ArticleNo{}

\begin{document}

\maketitle

\begin{abstract}
Given a string $\gtext{}$ of length $N$, 
the goal of grammar compression is to construct a small context-free grammar generating only $\gtext{}$.
Among existing grammar compression methods, RePair (recursive paring) [Larsson and Moffat, 1999] is notable
for achieving good compression ratios in practice.
Although the original paper already achieved a time-optimal algorithm to compute the RePair grammar $\repair(\gtext{})$ in expected $O(N)$ time,
the study to reduce its working space is still active so that it is applicable to large-scale data.
In this paper, we propose the first RePair algorithm working in compressed space, i.e., potentially $o(N)$ space for highly compressible texts.
The key idea is to give a new way to restructure an arbitrary grammar $\rg{}$ for $\gtext{}$ into $\repair(\gtext{})$ 
in compressed space and time.
Based on the recompression technique, 
we propose an algorithm for $\repair(\gtext{})$ in $O(\min(N, nm \log N))$ space and expected $O(\min(N, nm \log N) m)$ time 
or $O(\min(N, nm \log N) \log \log N)$ time,
where $n$ is the size of $\rg{}$ and $m$ is the number of variables in $\repair(\gtext{})$.
We implemented our algorithm running in $O(\min(N, nm \log N) m)$ time and show it can actually run in compressed space.
We also present a new approach to reduce the peak memory usage of existing RePair algorithms combining with our algorithms,
and show that the new approach outperforms, both in computation time and space, 
the most space efficient linear-time RePair implementation to date.
\end{abstract}

\section{Introduction}\label{sec:intro}
\subsection{Motivations and Contributions}
Given a string $\gtext{}$ of length $N$, 
the goal of grammar compression is to construct a small context-free grammar generating only $\gtext{}$.
Among existing grammar compression methods, RePair (recursive paring)~\cite{Larsson1999RePair} is notable
for achieving good compression ratios in practice and in theory~\cite{Navarro2008RePairAchievesHighOrderEntropy,Ganczorz2017EntropyBoundsForGrammarCompression}.
The principle of RePair is quite simple to explain:
it chooses one of the most frequent bigrams appearing in $\gtext{}$ more than once and
greedily replaces every occurrence of the bigram with a variable 
whose righthand side is the bigram,
and recursively applies the procedure to the resulting text until there is no bigram with frequency $\ge 2$.
This principle successfully captures the regularities frequently appearing in the text,
and so it has been shown that RePair (or the essence of RePair) has wide range of applications to, e.g.,
word-based text compression~\cite{Wan2003BrowsingAndSearchingCompressedDocuments},
compression of Web graphs~\cite{Claude2007FastCompactWebGraphRepresentation},
compressed suffix trees~\cite{Gonzalez2007CompressedTextIndexesWithFastLocate},
compressed wavelet trees~\cite{Navarro2011PCD},
tree compression~\cite{Lohrey2013Xts}, and
data mining~\cite{Tabei2016ScalablePartialLeastSquaresRegression_GrammarCompressedDataMatrices}.

In their original paper~\cite{Larsson1999RePair},
Larsson and Moffat proposed a time-optimal algorithm to compute the RePair grammar $\repair(\gtext{})$ in expected $O(N)$ time.
The space usage is analyzed as $5N + 4 \sigma^2 + 4 m + \ceil{\sqrt{N}}$ words,
where $\sigma$ is the alphabet size and $m$ is the number of variables in $\repair(\gtext{})$.
However, the space usage is not satisfying since the amount of data becomes larger and larger.
Thus, the study to reduce its working space is still active~\cite{BilleGortzPrezza2017,Bille2017PracticalEffectiveRePair}.

In this paper, we propose the first RePair algorithm working in compressed space, i.e., potentially $o(N)$ space for highly compressible texts.
The key idea is to give a new way to restructure an arbitrary grammar $\rg{}$ for $\gtext{}$ into $\repair(\gtext{})$ 
in compressed space and time.
More precisely, we show how to compute $\repair(\gtext{})$ in $O(\min(N, nm \log N))$ space and $O(\min(N, nm \log N) m)$ time,
and improve\footnote{to be precise, the improvement is achieved only when $m = \omega(\log\log N)$, which is likely to hold for compressible texts}
the expected time complexity to $O(\min(N, nm \log N) \log \log N)$,
where $n$ is the size of $\rg{}$ and $m$ is the number of variables in $\repair(\gtext{})$.
Note that $n$ and $m$ can be exponentially smaller than $N$, while $\log N \le n$.\footnote{$\log N \le m$ is not necessarily true since RePair stops producing variables when the input text is compressed into a string $w$ containing no bigram with frequency $\ge 2$. Still, it holds that $\log N \le m + |w|$.}

With our algorithms one can obtain $\repair(\gtext{})$ from $\gtext{}$ in compressed space as follows:
The input string is first processed by an \emph{online} grammar compression algorithm, such as~\cite{Takabatake2017Solca,Masaki2016OnlineRePair},
that works in compressed space, and then its output grammar is recompressed into $\repair(\gtext{})$.
This fits well the scenario in which data sources (such as embedded devices with sensors) have weaker computational resources,
and thus, the produced data is compressed by a lightweight compression algorithm (to reduce the transmission cost)
and sent to server in which further compression can be conducted.

Restructuring a compressed representation of data into another compressed representation \emph{in compressed space} has its own interest and applications,
and thus, has been widely studied.
In the seminal work~\cite{Charikar2005sgp,slplz} in the field of grammar compression,
restructuring LZ77~\cite{LZ77} into balanced grammars is the key to obtain a reasonable approximation to the smallest grammar.
In~\cite{Goto2011RCT}, a bunch of restructuring algorithms were considered in major lossless compression algorithms including
LZ77~\cite{LZ77}, LZ78~\cite{LZ78}, Bisection~\cite{Nelson1995Bisection}, and RePair~\cite{Larsson1999RePair}.
In~\cite{Bannai2012ELF,Bannai2013CSt}, the authors gave efficient algorithms to convert any grammar compressed string to LZ78.
Recently compressed space LZ77 parsing was achieved using another compressed scheme of
run-length compressed Burrows-Wheeler transform~\cite{2017PolicritiP_Lz77ComputBasedOnRun_Algorithmica,2018BannaiGI_OnlinLz77ParsinAndMatch_CPM}.
Our contribution in this paper is to draw a missing line from admissible grammars to the RePair grammar in Figure~1 of~\cite{Goto2011RCT}.
As pointed out in~\cite{Goto2011RCT,Bannai2012ELF}, restructuring has many applications,
e.g., dynamic updates of compressed strings and efficient computation of normalized compression distance (NCD)~\cite{Cilibrasi2005ClusterByCompression}.
As more and more data is available in compressed form, the importance of restructuring algorithms grows.

We implemented a prototype of our recompression algorithm for RePair with complexities of $O(\min(N, nm \log N))$ space and $O(\min(N, nm \log N) m)$ time.
While we confirm that it actually has a potential to run in compressed space, the running time is not fast enough to conduct comprehensive experiments over various datasets.
Instead of claiming the practicality of the current implementation,
we show some evidence that our $O(\min(N, nm \log N) \log \log N)$-time algorithm could be practical by further algorithmic engineering work.
In particular, our experimental results suggest that the $nm \log N$ term in the theoretical bounds could be loose,
and much smaller, say $O(n)$, for most of the cases in reality.
We also propose a new approach to reduce the peak memory usage of existing RePair algorithms combining with our method.
The experimental results show that the approach is promising, outperforming the most space efficient linear-time implementation
to date both in time and space.

\subsection{Related work.}
There have been several attempts to modify the original RePair grammar to improve its performance
in terms of working space~\cite{Wan2007BlockMergingForRePair,Sekine2014AdaptiveDicShare_RePair,Masaki2016OnlineRePair} and compression ratio~\cite{Ganczorz2017ImproveOnRePair}.

For the approximation ratio of RePair grammar to the smallest grammar generating the input string of length $N$,
Charikar et al.\ proved an upper bound $O((N / \log N)^{2/3})$ and lower bound $\Omega(\sqrt{\log N})$.
The lower bound was recently improved to $\Omega(\log N / \log \log N)$ in~\cite{Hucke2017ApproxRatioOfRePair}.

Our algorithms simulate the replacements of bigrams on grammars.
The technique used here is borrowed from the \emph{recompression} technique of Je\.z,
which has been proved to be a powerful tool in problems related to grammar compression~\cite{Jez2012CMN,Jez2015Aog,Jez2015FFC,Jez2014Aos,I2017LceWithRecompression}
and word equations~\cite{Jez2016OneVariableWordEquation_LinearTime,Jez2016Recompression_WordEquations}.
In particular, the grammar compression method based on recompression~\cite{Jez2015Aog} considers
replacing bigrams in a string with variables level by level like RePair.
The difference from RePair lies in the way of choosing bigrams to be replaced.
Instead of replacing the most frequent bigram in a single round,
recompression chooses several bigrams (which cannot overlap each other)
in a way that a given string shrinks by a constant factor after the round.
This strategy has lots of merits in theory, e.g., it assures that the number of rounds is $O(\log N)$
and the approximation ratio to the smallest grammar is $O(\log N)$, where $N$ is the length of an input string.
Moreover, the procedure is simulated from any grammar of size $n$ in $O(n \log^2(N/n))$ time
(or $O(n \log(N/n))$ time with a slight modification) and $O(n \log (N/n))$ space (see~\cite{I2017LceWithRecompression}).
The mechanism of replacing bigrams on grammars can also be used for RePair in a somewhat straightforward way.
As the way of choosing bigrams is different, we have to thoroughly reanalyze the complexities for RePair,
and as a result, unfortunately, we have lost the theoretical cleanness of recompression.
Still, RePair has a strong merit in practical compression ratio and
we show that our approach is helpful to overcome its weakness, the peak memory usage in compression.

\section{Preliminaries}\label{sec:prelim}

An alphabet $\alphabet$ is a finite set of symbols.
A string over $\alphabet$ is an element in $\alphabet^*$.
For any string $w \in \alphabet^{*}$, $|w|$ denotes the length of $w$.
Let $\emptystr$ be the empty string, i.e., $|\emptystr| = 0$.
Let $\alphabet^{+} = \alphabet^{*} \setminus \{ \emptystr \}$.
For any $1 \leq i \leq |w|$, $w[i]$ denotes the $i$-th symbol of $w$.
For any $1 \leq i \leq j \leq |w|$,
$w[i..j]$ denotes the substring of $w$ beginning at $i$ and ending at $j$.
For convenience, let $w[i..j] = \emptystr$ if $i > j$.
For any $0 \leq i \leq |w|$, $w[1..i]$ (resp.\ $w[|w|-i+1..|w|]$) is called the prefix (resp.\ suffix) of $w$ of length $i$.
We say that a string $x$ \emph{occurs} at the interval $[i..i+|x|-1]$ in $w$ iff $w[i..i+|x|-1] = x$.
A substring $w[i..j] = c^d~(c \in \alphabet, d \geq 1)$ of $w$ is called a \emph{block} iff it is a maximal run of a single symbol, 
i.e., $(i = 1 \vee w[i-1] \neq c) \wedge (j = |w| \vee w[j+1] \neq c)$.

An element $\lchar{c}\rchar{c}$ in $\alphabet^2$ is called a \emph{bigram}, and the bigram is said to be \emph{repeating} iff $\lchar{c} = \rchar{c}$.
When we mention the \emph{frequency} of a bigram $\lchar{c}\rchar{c}$ in $w \in \alphabet^*$,
it actually means the \emph{non-overlapping} frequency, which counts the maximum number of occurrences of $\lchar{c}\rchar{c}$ that do not overlap each other.
While the frequency of a non-repeating bigram is identical to the number of occurrences of $\lchar{c}\rchar{c}$,
the frequency of a repeating bigram is counted by summing up $\floor{d / 2}$ for every block $c^d$ of $c = \lchar{c} = \rchar{c}$.
Let $\freq(\lchar{c}\rchar{c}, w)$ denote the frequency of $\lchar{c}\rchar{c}$ in $w$.

The text subjected to being compressed is denoted by $\gtext{} \in \Sigma^{*}$ with $N = |\gtext{}|$ throughout this paper.
We assume that $\Sigma$ is an integer alphabet $[1..N^{O(1)}]$ and the standard word RAM model with word size $\Theta(\lg N)$.
The time complexities are expected time as RePair algorithms utilize hash functions to look-up/update frequency tables etc.
Also, the space complexities are measured by the number of words (not bits).

In this article, we deal with grammar compressed strings,
in which a string is represented by a Context-Free Grammar (CFG) generating the string only.
We simply use the term grammars or CFGs to refer to such specific CFGs for string compression.
In particular, we consider a normal form of CFGs, called \emph{Straight-Line Programs (SLPs)}, in which
the righthand side of every production rule is a bigram.\footnote{Of course, we ignore any trivial input string of length one or zero.}
Formally, an SLP that generates a string $\gtext{}$ is a triple
$\slp = (\Sigma_{\slp}, \vars_{\slp}, \rules_{\slp})$, where $\Sigma_{\slp}$ is the set of terminals (letters),
$\vars_{\slp}$ is the set of non-terminals (variables),
$\rules_{\slp}$ is the set of deterministic production rules whose righthand sides are in $(\vars_{\slp} \cup \Sigma_{\slp})^{2}$, and
the last variable derives $\gtext{}$.\footnote{We treat the last variable as the starting variable.}

For an SLP $\slp$ with $n = |\vars_{\slp}|$,
note that $N$ can be as large as $2^{n}$, and so, 
SLPs have a potential to achieve exponential compression.
Also, $n \ge \lg N$ is always true.
We treat variables as integers in $[1..n]$ 
(which should be distinguishable from $\Sigma_{\slp}$ by having one extra bit), 
and $\rules_{\slp}$ as an injective function that maps a variable to its righthand side
(i.e., $\rules_{\slp}(X)$ represents a bigram for any $X \in \vars_{\slp}$).
For any $X \in \vars_{\slp}$, if $\rules_{\slp}(X)[1]$ (resp. $\rules_{\slp}(X)[2]$) is from $\vars_{\slp}$, it is called the left (resp.\ right) variable of $X$.
Let $\dtree_{\slp}$ denote the derivation tree of $\slp$.
Note that $\dtree_{\slp}$ is implicitly stored by the production rules in $O(n)$ space, which can be seen as a DAG representation of the tree.
We assume that variables are in a (reversed) topological sort order, i.e., left/right variable of $X$ is smaller than $X$.
Let $\vocc_{\slp}(X)$ denote the number of nodes labeled with $X$ in $\dtree_{\slp}$.
It is a well-known fact that we can preprocess $\slp$ in $O(n)$ time and space to compute $\vocc_{\slp}(X)$ for all $X \in \vars_{\slp}$
by a simple dynamic programming (it reduces to the problem of computing the number of paths from the source to nodes in a DAG).
We assume that given any variable $X$ we can access in $O(1)$ time the information on $X$, e.g., $\rules_{\slp}(X)$ and $\vocc_{\slp}(X)$.
For any variable $X \in \vars$, the string derived from $X$ is denoted by $\val{\slp}(X)$, where we omit $\slp$ when it is clear from context.

RePair~\cite{Larsson1999RePair} is a grammar compression algorithm, 
which recursively replaces the most frequent bigram (tie-breaking arbitrary) into a variable
while there is a bigram with frequency $\ge 2$.
Formally, RePair transforms $\gtext{0} := \gtext{}$ level by level into strings, $\gtext{1}, \gtext{2}, \dots, \gtext{m}$:
at the $h$-th level ($0 \le h$) we are given $\gtext{h}$ and compute $\gtext{h+1}$ 
that is obtained by replacing $\freq(\lchar{c}\rchar{c}, \gtext{h})$ non-overlapping occurrences of 
the most frequent bigram $\lchar{c}\rchar{c}$ in $\gtext{h}$ with
a new variable $\hat{c}$ such that $\hat{c} \rightarrow \lchar{c}\rchar{c}$.
To remove ambiguity in the replacement for a repeating bigram $\lchar{c}\rchar{c}$ with $\lchar{c} = \rchar{c} = c$,
let us conduct a greedy left-to-right parsing on a block $c^d$, namely,
$c^d$ is replaced with $\hat{c}^{\floor{d/2}}$ if $d$ is even, and otherwise $\hat{c}^{\floor{d/2}}c$.
Any appearance $\hat{c}$ in $\gtext{h}$ is treated as a letter in the later rounds,
so we call variable $\hat{c}$ the letter introduced at level $h+1$.
The process shrinks the string monotonically, and finally we get $\gtext{m}$ in which
there are no bigram with frequency $\ge 2$.

Let $\repair(\gtext{})$ denote the grammar obtained by RePair with input $\gtext{}$.
The variables of $\repair(\gtext{})$ consist of the letters introduced at all levels
and the starting variable whose righthand side is $\gtext{m}$.
Except the starting variable, the righthands of the rules are bigrams.

\section{$O(\min(N, nm \log N) m)$-time algorithm}\label{sec:algo_scan}
In this section we show how, given an arbitrary SLP $\slp$ generating $\gtext{}$,
we compute $\repair(\gtext{})$ in $O(\min(N, nm \log N) m)$ time and $O(\min(N, nm \log N))$ space,
where $N$ is the length of $\gtext{}$, and $n$ (resp. $m$) is the number of variables in $\slp$ (resp. $\repair(\gtext{})$).

\subsection{Overview: Recompress $\slp$ into $\repair(\gtext{})$ in compressed space.}\label{sec:overview}

The key idea to compute $\repair(\gtext{})$ in compressed space is to \emph{recompress}
an arbitrary $\slp$ for $\gtext{}$ into $\repair(\gtext{})$ without decompressing $\slp$.
For a clear description, we add two auxiliary variables
that introduce sentinels at the beginning\footnote{we assign index zero to $\#$ so that the indexes in $\gtext{}$ are persistent with the original ones}
$T[0] = \# \notin \letters{}$ and at the end $T[N+1] = \$ \notin \letters{}$:
we define $\rg{0} := (\letters{0}, \hvars{0}, \hrules{0})$ such that 
$\letters{0} := \letters{} \cup \{ \#, \$ \}$,
$\hvars{0} := \vars \cup \{ X_{\#}, X_{\$} \}$, and
$\hrules{0} := \hrules{} \cup \{ (X_{\#} \rightarrow \# X_{n}), (X_{\$} \rightarrow X_{\#} \$) \}$,
where $X_{n}$ is the starting variable of $\slp$.
Clearly, $\rg{0}$ generates $\#\gtext{}\$ $.

We employ the recompression technique~\cite{Jez2012CMN,Jez2015Aog,Jez2015FFC,Jez2014Aos}, invented by Je\.z,
to simulate the transformation from $\gtext{h-1}$ to $\gtext{h}$ on CFGs.
We transform level by level $\rg{0}$ into a sequence of CFGs, 
$\rg{1} = (\letters{1}, \hvars{0}, \hrules{1}), \rg{2} = (\letters{2}, \hvars{0}, \hrules{2}), \dots, \rg{m} = (\letters{m}, \hvars{0}, \hrules{m})$,
where each $\rg{h}$ generates $\#\gtext{h}\$ \in \letters{h}^*$.
Namely, compression from $\gtext{h}$ to $\gtext{h+1}$ is simulated on $\rg{h}$.
We can correctly compute the letters introduced at each level $h+1$ while modifying $\rg{h}$ into $\rg{h+1}$,
and hence, we get all the letters of $\repair(\gtext{})$ in the end.
We note that new variables for $\rg{h}$ are never introduced and 
the modification is done by rewriting righthand sides of the original variables in $\hvars{0}$.
During the modification, the string represented by a variable $X$ could be shorten,
and $X$ could be $\NULL$ meaning that it represents nothing, i.e., $\val{\rg{h}}(X) = \emptystr$.

Here we introduce the special formation of the CFGs $\rg{h}$ (it is a generalization of SLPs):
For any $X \in \hvars{0}$, $\hrules{h}(X)$ consists of an arbitrary number of letters and at most two non-null variables
that are originally in $\hrules{0}(X)$.
More precisely, the following condition holds:
\begin{list}{}{}
\item For any variable $X \in \hvars{0}$, let $\lchar{X}$ (resp. $\rchar{X}$) denote the left (resp.\ right) variable, where
      it represents $\NULL$ if it does not exist.
      Then, $\hrules{h}(X) = \lchar{X} w_{X} \rchar{X}$ with $w_{X} \in \letters{h}^*$, where
      null variables are imaginary and actually removed from $\hrules{h}(X)$.
\end{list}
In addition, we compress $w_{X}$ by the run-length encoding 
so that it can be stored in $O(\rlesize{w_{X}})$ space, where $\rlesize{w_{X}}$ denotes the number of blocks in $w_{X}$.
We define the size of $\hrules{h}(X)$ by $\rlesize{w_{X}}$ plus the number of non-null variables in $\hrules{h}(X)$,
and denote it by $\rlesize{\hrules{h}(X)}$.
The size of $\rg{h}$, denoted by $|\rg{h}|$, is defined by $\sum_{X \in \hvars{0}} \rlesize{\hrules{h}(X)}$.

In Subsection~\ref{sec:freq}, we show how to compute the frequencies of bigrams on $\rg{h}$ in $O(|\rg{h}|)$ time and space.
In Subsection~\ref{sec:replace}, we show, given the most frequent bigram $\lchar{c}\rchar{c}$, how to replace $\lchar{c}\rchar{c}$ with a new letter $\hat{c}$ on $\rg{h}$ to get $\rg{h+1}$ in $O(|\rg{h}|)$ time and space.
In Subsection~\ref{sec:analysis}, we show that $|\rg{h}| = O(\min(N, nh \log N))$ for any level $h$, and thus,
the recompression from $\rg{0}$ to $\rg{m}$ can be done in the claimed time and space complexity.

\subsection{How to compute frequencies of bigrams on $\rg{h}$.}\label{sec:freq}

The goal of this subsection is to show the next lemma:
\begin{lemma}\label{lem:freq}
  Given $\rg{h}$ generating $\gtext{h}$,
  we can compute in $O(|\rg{h}|)$ time and space the frequencies of bigrams appearing in $\gtext{h}$.
\end{lemma}

The following fact is useful to compute the frequencies of bigrams in $\gtext{h}$ on $\rg{h}$.
\begin{fact}\label{fact:kushi}
  For any interval $[i..j] \subseteq [0..|\gtext{h}| + 1]$ with $j-i > 0$, there is a unique variable $X \in \hvars{0}$ that
  is the label of the lowest common ancestor of the $i$-th and $j$-th leaf in $\dtree_{\rg{h}}$.
  We say that such $X$ \emph{stabs} $[i..j]$.
\end{fact}
According to Fact~\ref{fact:kushi}, we can detect the occurrences of bigrams
by variables that stab the occurrences without duplication or omission.
In addition, since each variable $X$ can stab at most $\rlesize{\hrules{h}(X)}$ distinct bigrams,
it implies that there are at most $\sum_{X \in \hvars{0}} \rlesize{\hrules{h}(X)} = |\rg{h}|$ distinct bigrams in total.

In order to compute the frequencies, we use the following auxiliary information for all variables,
which can be computed in a bottom-up manner in $O(|\rg{h}|)$ time and stored in $O(n)$ space.
\begin{itemize}
\item $\lmb{X}$: the leftmost block in $\val{\rg{h}}(X)$.
\item $\rmb{X}$: the rightmost block in $\val{\rg{h}}(X)$.
\item $\isSB{X}$: Boolean that represents if $\hrules{h}(X)$ consists of a single block.
\end{itemize}
For any variable $X \in \hvars{0}$ with $\hrules{h}(X) = \lchar{X} w_{X} \rchar{X}$,
we can easily compute $\lmb{X}$, $\rmb{X}$ and $\isSB{X}$ in $O(1)$ time,
assuming that we have computed those for $\lchar{X}$ and $\rchar{X}$:
for example, $\lmb{X}$ is identical to $\lmb{\lchar{X}}$ if the prefix block stops inside $\lchar{X}$,
or it is extended if $\lmb{\lchar{X}}$ can be merged with the first block of $w_{X}$ (and further with $\lmb{\rchar{X}}$).

We first focus on the frequencies of non-repeating bigrams $\lchar{c}\rchar{c}$.
According to Fact~\ref{fact:kushi}, we assign any occurrence $[i..i+1]$ of $\lchar{c}\rchar{c}$ to the variable that stabs $[i..i+1]$
without duplication or omission.
We now intend to count all the occurrences of $\lchar{c}\rchar{c}$ assigned to $X$ in $\hrules{h}(X) := \lchar{X} w_X \rchar{X}$.
Observe that $\lchar{c}\rchar{c}$ appears \emph{explicitly} in $w_X$ or \emph{crosses} the boundaries of $\lchar{X}$ and/or $\rchar{X}$.
Thus, it is enough to compute the frequencies in $\rmb{\lchar{X}} w_{X} \lmb{\rchar{X}}$.
Since each $\lchar{c}\rchar{c}$ found in $\rmb{\lchar{X}} w_{X} \lmb{\rchar{X}}$ appears 
every time a node labeled with $X$ appears in $\dtree_{\rg{h}}$, we count each occurrence of $\lchar{c}\rchar{c}$ in $\rmb{\lchar{X}} w_{X} \lmb{\rchar{X}}$ with the weight $\vocc(X)$.
Hence, the frequencies of non-repeating bigrams can be computed in $O(|\rg{h}|)$ time
while scanning $\rmb{\lchar{X}} w_{X} \lmb{\rchar{X}}$ for all $X \in \hvars{0}$ and 
incrementing the frequency of $\lchar{c}\rchar{c}$ by $\vocc(X)$
whenever we find an occurrence of a non-repeating bigram $\lchar{c}\rchar{c}$ in $\rmb{\lchar{X}} w_{X} \lmb{\rchar{X}}$.

Next we compute the frequencies of repeating bigrams.
To this end, we detect all the blocks with lengths $\ge 2$ without duplication or omission
by assigning each block to the smallest variable that ``witnesses'' the maximality of the block.
Formally, we assign a block occurring at $[i..j]$ to the variable $X$ that stabs $[i-1..j+1]$.
(Note that $[i-1..j+1]$ is always a valid interval thanks to the sentinels $\#$ and $\$ $.)
For any $X$ with $\hrules{h}(X) := \lchar{X} w_{X} \rchar{X}$,
we can find every block assigned to $X$ as a block appearing in $\rmb{\lchar{X}} w_{X} \lmb{\rchar{X}}$,
where we ignore a block that is a prefix/suffix of $\val{\rg{h}}(X)$ because $X$ does not witness its maximality.
Using the information of $\isSB{\lchar{X}}$ and $\isSB{\rchar{X}}$,
we can easily check if a block is a prefix/suffix of $\val{\rg{h}}(X)$.
The frequencies of repeating bigrams can be computed in $O(|\rg{h}|)$ time
while scanning $\rmb{\lchar{X}} w_{X} \lmb{\rchar{X}}$ for all $X \in \hvars{0}$ and 
incrementing the frequency of $c^2$ by $\floor{d/2} \vocc(X)$
whenever we find a block $c^d$ with $d \ge 2$ that is assigned to $X$.

Figure~\ref{fig:compute_freq} shows an example on how to compute the frequencies on grammars.
\begin{figure}[t]
\begin{center}
  \includegraphics[scale=0.26]{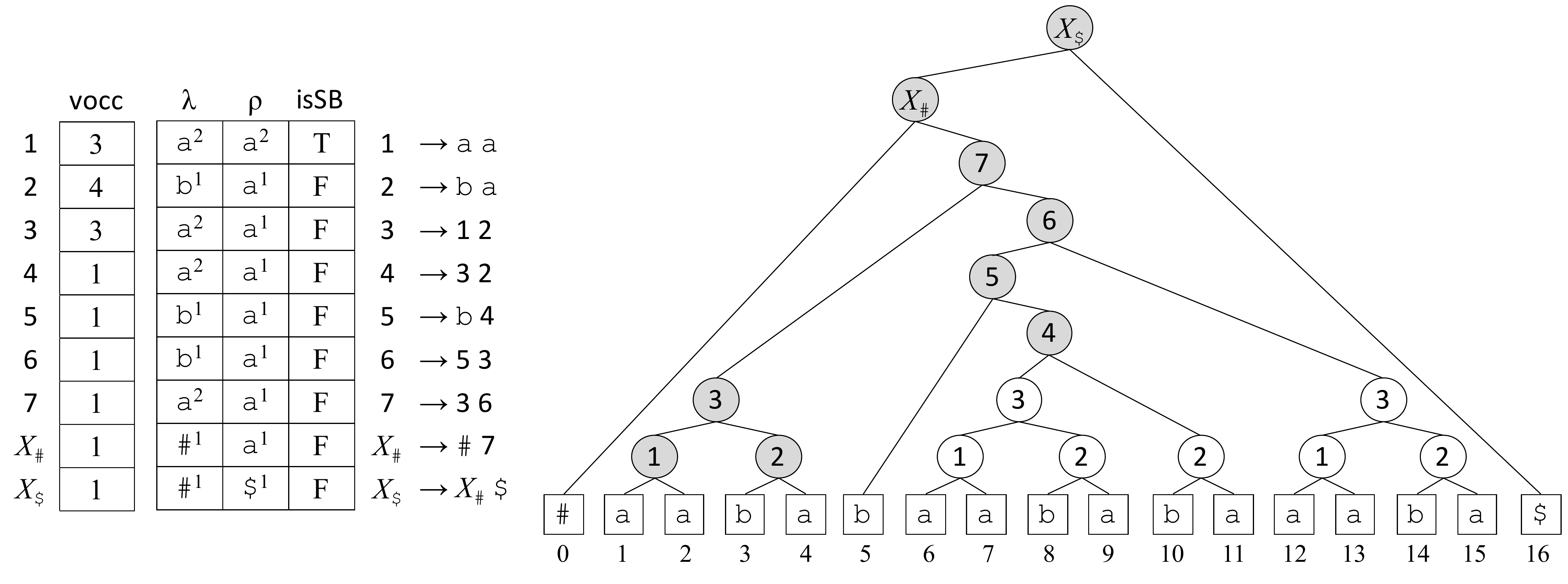}
  \caption{
    An example of $\rg{0}$ and its derivation tree are depicted, where the variables are numbered in post-order and the first appearance of each variable is shaded.
    Using $\vocc$ table, and the information of $\lmb{\cdot}$, $\rmb{\cdot}$ and $\isSB{\cdot}$,
    we can compute the frequencies of each bigram as follows:
    An occurrence of non-repeating bigram $\texttt{ab}$ is stabbed by the variables $\texttt{3}, \texttt{4}$ and $\texttt{7}$, and so,
    the frequency of $\texttt{ab}$ is $1 \times \vocc(\texttt{3}) + 1 \times \vocc(\texttt{4}) + 1 \times \vocc(\texttt{7}) = 3 + 1 + 1 = 5$.
    An occurrence of non-repeating bigram $\texttt{ba}$ is stabbed by the variables $\texttt{2}$ and $\texttt{5}$, and so,
    the frequency of $\texttt{ba}$ is $1 \times \vocc(\texttt{2}) + 1 \times \vocc(\texttt{5}) = 4 + 1 = 5$.
    Finally, since $\texttt{5}$, $\texttt{6}$ and $X_{\#}$ stab intervals which witness the maximality of blocks of $\texttt{a}$'s,
    they count the frequency of $\texttt{aa}$ by
    $\floor{2/2} \times \vocc(\texttt{5}) + \floor{3/2} \times \vocc(\texttt{6}) + \floor{3/2} \times \vocc(X_{\#}) = 1 + 1 + 1 = 3$.
    Note that the variable $\texttt{1}$ stabs $\texttt{aa}$, but it counts nothing as it does not witness the maximality of the block.
  }
\label{fig:compute_freq}
\end{center}
\end{figure}

\subsection{How to transform $\rg{h}$ into $\rg{h+1}$.}\label{sec:replace}

The goal of this subsection is to show the next lemma:
\begin{lemma}\label{lem:replace}
  Given $\rg{h}$ generating $\gtext{h}$ and the most frequent bigram $\lchar{c}\rchar{c}$ in $\gtext{h}$,
  we can transform $\rg{h}$ into $\rg{h+1}$ in $O(|\rg{h}|)$ time and space.
\end{lemma}

We first focus on the case where $\lchar{c}\rchar{c}$ is non-repeating.
Some of the occurrences of $\lchar{c}\rchar{c}$ are explicitly written in $w_{X}$ 
and the others are crossing the boundaries of left and/or right variables of $X$ for some $X \in \hvars{0}$.
While explicit occurrences can be replaced easily, crossing occurrences need additional treatment.
To deal with crossing occurrences, we first \emph{uncross} them by popping out every $\lchar{c}$ (resp. $\rchar{c}$)
occurring at the rightmost (resp.\ leftmost) position of $\val{\rg{h}}(Y)$ and popping them into the appropriate positions in the other rules.
More precisely, we do the following ``simultaneously'' for all $X \in \hvars{0}$:
\begin{description}
\item[$\pil$]
      If $\hrules{h}(X)$ contains a variable $Y \in \hvars{0}$ in any position other than the first position and $\val{\rg{h}}(Y)[1] = \rchar{c}$, 
      replace the occurrence of $Y$ with $\rchar{c} Y$; and
      if $\hrules{h}(X)$ contains a variable $Y \in \hvars{0}$ in any position other than the last position and $\val{\rg{h}}(Y)[|\val{\rg{h}}(Y)|] = \lchar{c}$, 
      replace the occurrence of $Y$ with $Y \lchar{c}$.
\item[$\pol$]
      If $\hrules{h}(X)[1] = \rchar{c}$, delete it; and
      if $\hrules{h}(X)[|\hrules{h}(X)|] = \lchar{c}$, delete it.
      In addition, if $X$ becomes $\NULL$, we remove all the occurrences of $X$ in $\hrules{h}$.
\end{description}
$\pol$ removes $\lchar{c}$ (resp.\ $\rchar{c}$) from the rightmost (resp.\ leftmost) position of $\val{\rg{h}}(Y)$ (which can be a part of a crossing occurrence of $\lchar{c}\rchar{c}$), 
and $\pil$ introduces the removed letters into appropriate positions in $\hrules{h}$ so that the modified $\rg{h}$ keeps to generate $\gtext{h}$.
The uncrossing can be conducted in $O(|\rg{h}|+n)$ time using the information of $\lmb{\cdot}$ and $\rmb{\cdot}$.
Since all the occurrences of $\lchar{c}\rchar{c}$ are now explicitly written in the righthand sides,
we can easily replace them with a fresh letter $\hat{c}$
while scanning the righthand sides in $O(|\rg{h}| + n)$ time.

Figure~\ref{fig:replace} shows an example on how the replacements of the first level is done on the grammar in Figure~\ref{fig:compute_freq}.
\begin{figure}[t]
\begin{center}
  \includegraphics[scale=0.26]{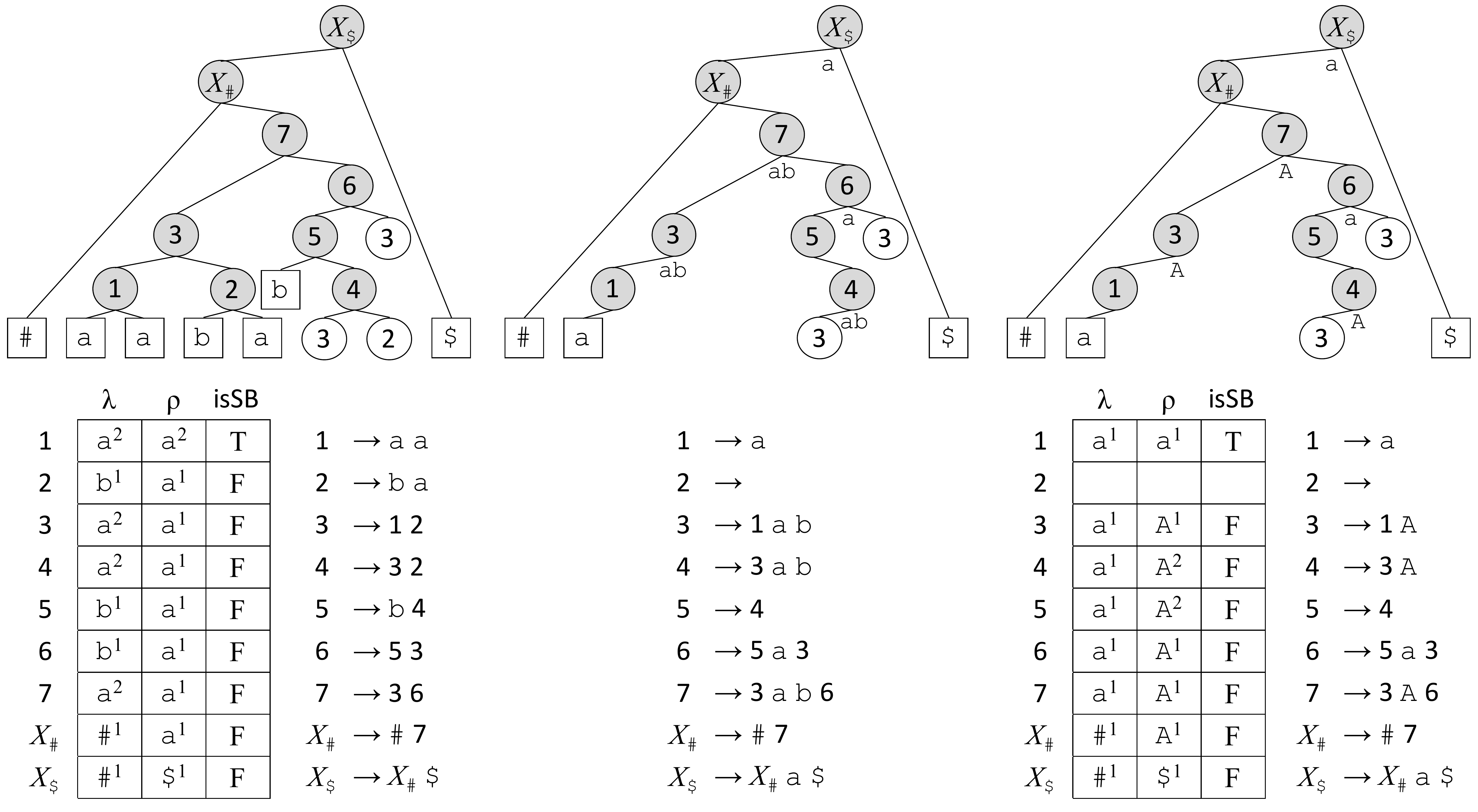}
  \caption{
    A demonstration of the replacements of a non-repeating bigram $\texttt{ab}$ on the grammar in Figure~\ref{fig:compute_freq}.
    The figure in the left depicts the grammar and tables in Figure~\ref{fig:compute_freq},
    where the subtrees under the unshaded nodes are suppressed for presentation.
    The figure in the middle shows an intermediate grammar just after we conduct $\pil$ and $\pol$.
    The figure in the right shows the resulting grammar after replacing every occurrence of $\texttt{ab}$ with $\texttt{A}$,
    and updated tables for the next level.
  }
\label{fig:replace}
\end{center}
\end{figure}

Next we deal with the case $\lchar{c}\rchar{c}$ is a repeating bigram, i.e., $c = \lchar{c} = \rchar{c}$.
We consider the blocks $c^d$ with $d \ge 2$ assigned to $X \in \hvars{0}$,
which can be found in $\rmb{\lchar{X}} w_{X} \lmb{\rchar{X}}$.
In a similar way to the non-repeating case, we first uncross $c^d$ if it starts in $\rmb{\lchar{X}}$ or ends in $\lmb{\rchar{X}}$.
The uncrossing for all variables can be done in $O(n)$ time and space.

\subsection{Analysis.}\label{sec:analysis}
The primal goal of this subsection is to prove Lemma~\ref{lem:size},
which upper bounds the CFG sizes during modification.

\begin{lemma}\label{lem:size}
  For any level $h$, $|\rg{h}| = O(\min(N, nh \log N))$.
\end{lemma}
\begin{proof}
  When transforming $\rg{h}$ into $\rg{h+1}$, there are two situations where the size of the righthand sides increases:
  (1) when letters/blocks are popped in; and (2) when a repeating bigram $cc$ is replaced on a run-length encoded block $c^d$ with odd $d > 2$.
  For (1), it is easy to see that for each variable $X$ the positions where letters/blocks popped in is at most two (the boundaries of left/right variables),
  and thus, the size of $\rg{h}$ increases at most $2(n+2) = O(n)$ for each level.
  For (2), we deposit $\log d \le \log N$ credit whenever a block $c^d$ is popped into some position so that the later increase by case (2) can be paid from the credit.
  Since at most $O(n \log N)$ credit is issued for each level, we obtain the bound $|\rg{h}| = O(nh \log N)$.
  Also, the number of occurrences of letters in the righthand sides of $\hrules{h}$ cannot be larger than the uncompressed size $|\gtext{h}|$,
  and therefore, $|\rg{h}| \le |\gtext{h}| + 2n = O(N)$ holds.
\end{proof}

Our first algorithm running in $O(\min(N, nm \log N) m)$ time and $O(\min(N, nm \log N))$ space is immediately obtained from 
Lemmas~\ref{lem:freq},~\ref{lem:replace} and~\ref{lem:size}.
\begin{theorem}\label{theorem:algo1}
  Given an SLP $\rg{}$ generating $\gtext{}$ of length $N$,
  we can compute $\repair(\gtext{})$ in $O(\min(N, nm \log N) m)$ expected time and $O(\min(N, nm \log N))$ space,
  where $n$ and $m$ are the numbers of variables in $\rg{}$ and $\repair(\gtext{})$, respectively.
\end{theorem}
\begin{proof}
  We first compute $\vocc_{\slp}(X)$ for all $X \in \vars_{\slp}$ in $O(n)$ time and space.
  At any level $h~(0 \le h < m)$, the transform from $\gtext{h}$ to $\gtext{h+1}$ is simulated on CFGs as follows:
  Given $\rg{h}$ generating $\gtext{h}$,
  we use Lemma~\ref{lem:freq} to compute the most frequent bigram in $\gtext{h}$,
  and Lemma~\ref{lem:replace} to obtain $\rg{h+1}$ that generates $\gtext{h+1}$.
  It can be done in $O(|\rg{h}|)$ time and space.
  Since $|\rg{h}| = O(\min(N, nh \log N))$ due to Lemma~\ref{lem:size},
  we can go through from $\rg{0}$ to $\rg{m}$ in $O(\sum_{h = 0}^{m} |\rg{h}|) = O(\min(N, nm \log N) m)$ time and 
  $O(\max \{|\rg{h}| \mid 0 \le h \le m\}) = O(\min(N, nm \log N))$ space.
\end{proof}

We note that the bound $|\rg{h}| = O(nh \log N)$ of Lemma~\ref{lem:size} could be quite rough
because the analysis considers the following (probably too pessimistic) scenario:
there are $\Omega(h)$ levels at which $\Omega(n)$ run-compressed letters are popped in and 
each of them produces $\Omega(\log N)$ remainders during replacing repeating bigrams on it.
In addition, the analysis does not take into account the fact that each replacement on non-repeating bigrams
reduces the grammar size by one.
It is open if there is an example to achieve the upper bound.
In our preliminary experiments, we observed that $|\rg{h}|$ is just a few times larger than $n$ in highly repetitive datasets.

\section{$O(\min(N, nm \log N) \log\log N)$-time algorithm}\label{sec:algo_list}

In this section, we improve the time complexity $O(\min(N, nm \log N) m)$ of Theorem~\ref{theorem:algo1} to $O(\min(N, nm \log N) \log \log N)$.
It is analogue to improving a naive $O(Nm)$-time RePair algorithm that works on plain text $\gtext{}$ to an $O(N)$-time algorithm.
At level $h~(0 \le h < m)$, the naive algorithm simply scans text $\gtext{h}$ to compute the most frequent
bigram and replace its non-overlapping occurrence with a fresh letter spending $O(|\gtext{h}|)$ time, and thus,
it takes $O(\sum_{h = 0}^{m} |\gtext{h}|) = O(Nm)$ time in total.
The essential idea of~\cite{Larsson1999RePair} to obtain $O(N)$-time algorithm is to:
\begin{enumerate}
\item represent $\gtext{h}$ by a linked list so that replacements can be done locally without breaking adjacent letters apart,
\item maintain, for every bigram in $\gtext{h}$, pointers to traverse all and only the occurrences of the bigram,
\item maintain the frequencies of all bigrams in a priority queue.
\end{enumerate}
At each level $h$, we obtain the most frequent bigram $\lchar{c}\rchar{c}$ from the priority queue and
replace every occurrence of $\lchar{c}\rchar{c}$ using the pointers to visit the occurrences of $\lchar{c}\rchar{c}$.
While replacing each occurrence, we can easily update the linked-list, pointers and frequencies of bigrams 
that are affected by the replacement in constant time.
Since the total number of replacement is at most $N$, the algorithm runs in $O(N)$ time.

We apply this idea to our algorithm in Section~\ref{sec:algo_scan}:
we maintain the linked-list for each righthand side and
pointers to traverse all and only the occurrences of any bigram appearing in the grammar
(it is explicitly written in the grammar rules or crossing the boundaries of left/right variables).
Here updating the information for bigrams crossing the boundaries is sometimes problematic
as the leftmost/rightmost descendants who possess the contexts beyond boundaries dynamically change.
We do not see how we can efficiently maintain it along with replacements, 
but at least we can recollect, for each level $h$, the information by computing 
$\lmb{\cdot}$, $\rmb{\cdot}$ and $\isSB{\cdot}$ in $O(n_{h})$ time
(as we did in the algorithm in Section~\ref{sec:algo_scan}),
where $n_{h}$ is the number of non-null variables in $\rg{h}$.

Note that in the algorithm working on uncompressed texts,
the priority queue can be implemented by a simple linked-list because every single replacement
increases/decreases the frequency of a bigram by one, and we can afford to spend the cost of maintaining the list
to run in $O(N)$ time.
However, this is not satisfiable for our ``compressed-time'' algorithm, which potentially runs in $o(N)$ time.
Thus, we use dynamic data structure for predecessor queries to implement the priority queue.
For example, using the y-fast trie~\cite{Willard1983LWR} we can update the frequency of a bigram in $O(\log \log N)$ expected time
while supporting the function of the priority queue in $O(\log \log N)$ time as well.
Then the algorithm runs in $O(\sum_{h = 0}^{m} n_{h} + R \log \log N)$ time and $O(n + R)$ space,
where $R$ is the total number of replacements executed on the grammars in our algorithm.
Since $R = O(\min(N, nm \log N))$ by Lemma~\ref{lem:size}, we can get the following theorem:

\begin{theorem}\label{theorem:algo2}
  Given an SLP $\rg{}$ generating $\gtext{}$ of length $N$,
  we can compute $\repair(\gtext{})$ in expected $O(\min(N, nm \log N) \log \log N)$ time and $O(\min(N, nm \log N))$ space,
  where $n$ and $m$ are the numbers of variables in $\rg{}$ and $\repair(\gtext{})$, respectively.
\end{theorem}

\section{Experiments}\label{sec:exp}
In this section, we show the results of our preliminary experiments.
We implemented in C++ our algorithm to compute $\repair(\gtext{})$ from an arbitrary grammar $\rg{}$ for $\gtext{}$
running in $O(\min(N, nm \log N) m)$ expected time and $O(\min(N, nm \log N))$ space.

We choose the following three highly repetitive texts in repcorpus,
\texttt{einstein.en.txt} (446 MB), \texttt{world\_leaders} (45 MB) and \texttt{fib41} (255 MB).\footnote{See \url{http://pizzachili.dcc.uchile.cl/repcorpus/statistics.pdf} for statistics of the datasets.}
We first compress each dataset by SOLCA~\cite{Takabatake2017Solca},
a space-optimal online grammar compression, to obtain $\rg{}$, and feed $\rg{}$ to our algorithm.
In theory, SOLCA runs in $O(N \log \log n)$ time and $O(n)$ space.

\begin{table}[t]
\begin{center}
  \begin{tabular}{| r || r | r | r | r | r | r | r | r |} \hline
                    &   time &  space &           $n$ &           $m$ &           Max & $\sum_{h = 0}^{m}|\rg{h}|$ & $\sum_{h = 0}^{m} n_{h}$ &           $R$ \\
                    &    [s] &   [MB] & $2^{10}\times$ & $2^{10}\times$ & $2^{10}\times$ &            $2^{10}\times$ &        $2^{10}\times$ & $2^{10}\times$ \\ \hline \hline
    einstein.en.txt & 5,626 & 27.36 & 413 & 98 & 1,157 & 38,408,764 &  11,149,315 & 5,241 \\ \hline
    world\_leaders & 19,872 & 33.07 & 807 & 204 & 1,920 & 139,854,080 & 36,212,346  & 14,249 \\  \hline 
    fib41 & 20 & 24.01 & 0.4 & 0.04 & 0.2 & 3 & 3 & 6,495 \\ \hline
  \end{tabular}
 \end{center}
 \caption{
   Table showing time and working space for our algorithm to compute RePair from each dataset (including the time and space of SOLCA).
   SOLCA takes $53$, $9$ and $19$ seconds for each dataset.
   The peak memory usage of \texttt{fib41} is from the constant-size hash table used in SOLCA.\@
   For other columns, $n$ is the number of variables in the output grammar $\rg{}$ of SOLCA,
   $m$ is the number of variables in the RePair grammar,
   Max $:= \max \{|\rg{h}| \mid 0 \le h \le m\}$ and
   $R$ is the total number of replacements executed on the grammars in the algorithm.
 }
 \label{table:result}
\end{table}

Table~\ref{table:result} summarizes the results, 
where we also collected some data during the execution, which are useful for understanding the performance.
The running time and working space of our algorithm deeply depend on the compressibility of each dataset.
We confirmed that our algorithm potentially runs in compressed space for repetitive texts.
We see that the recompression part for the extremely compressible text \texttt{fib41} is done in a second.
Unfortunately, for less compressible datasets our implementation does not scale well as $n$ and $m$ become larger.
More precisely, the running time of our algorithm depends on $\sum_{h = 0}^{m}|\rg{h}|$, i.e.,
our algorithm runs in $\Theta(\sum_{h = 0}^{m}|\rg{h}|)$ time.
As the value $\sum_{h = 0}^{m}|\rg{h}|$ is large even for relatively compressible datasets we tested,
it may be hopeless to make the algorithm practical.

As mentioned in Section~\ref{sec:algo_list}, our second algorithm runs
in $O(\sum_{h = 0}^{m} n_{h} + R \log \log N)$ time and $O(n + \max \{|\rg{h}| \mid 0 \le h \le m\})$ space,
where $n_{h}$ is the number of non-null variables in $\rg{h}$ and
$R$ is the total number of replacements executed on the grammars in the algorithm.
Because $R$ is upper bounded by $N$, the term $R \log \log N$ is almost linear in the worst-case.
As we see Table~\ref{table:result}, $R$ is actually much smaller than $N$.
Also, Table~\ref{table:result} shows that $\sum_{h = 0}^{m} n_{h}$ is not so big compared to $\sum_{h = 0}^{m}|\rg{h}|$,
and thus, we expect that our second algorithm runs in a reasonable time.

Next we propose a new approach to reduce the peak memory usage of existing algorithms by combining with our algorithms.
Since the peak memory usage is achieved at the very beginning of RePair, we can avoid it as follows:
introducing paramter $t$, we first use our algorithms until the input text $\gtext{}$ becomes sufficiently small, i.e.,
$|\gtext{h}| < |\gtext{}| / t$, and then, switch to a linear time algorithm that works in $O(|\gtext{h}|)$ time and space.
In our experiments, we combine our implementation described above with
a well-tuned implementation of linear-time RePair by Maruyama~\cite{MaruyamaRP} (denote it by \texttt{RP}).
Setting $t \in \{2, 3, 4, 5\}$, we compare our method with \texttt{RP}
and the most space efficient linear-time algorithm~\cite{Bille2017PracticalEffectiveRePair,SERP} to date (denote it by \texttt{SERP}).
In theory, \texttt{SERP} runs in $O(N/\epsilon)$ time using at most $(1.5 + \epsilon)N$ words of space for arbitrary small $\epsilon \le 1$,
but $\epsilon$ is fixed to $1$ in their implementation.
The results for some datasets from repcorpus are shown in Figure~\ref{fig:graph}.
We can see that our approach successfully slashes the peak memory usage of \texttt{RP}.
Also, the time-space tradeoff is controled by parameter $t$ and
our method with $t = 3$ outperforms \texttt{SERP} both in time and space.

\begin{figure}[t]
  \begin{center}
    \includegraphics[scale=0.62]{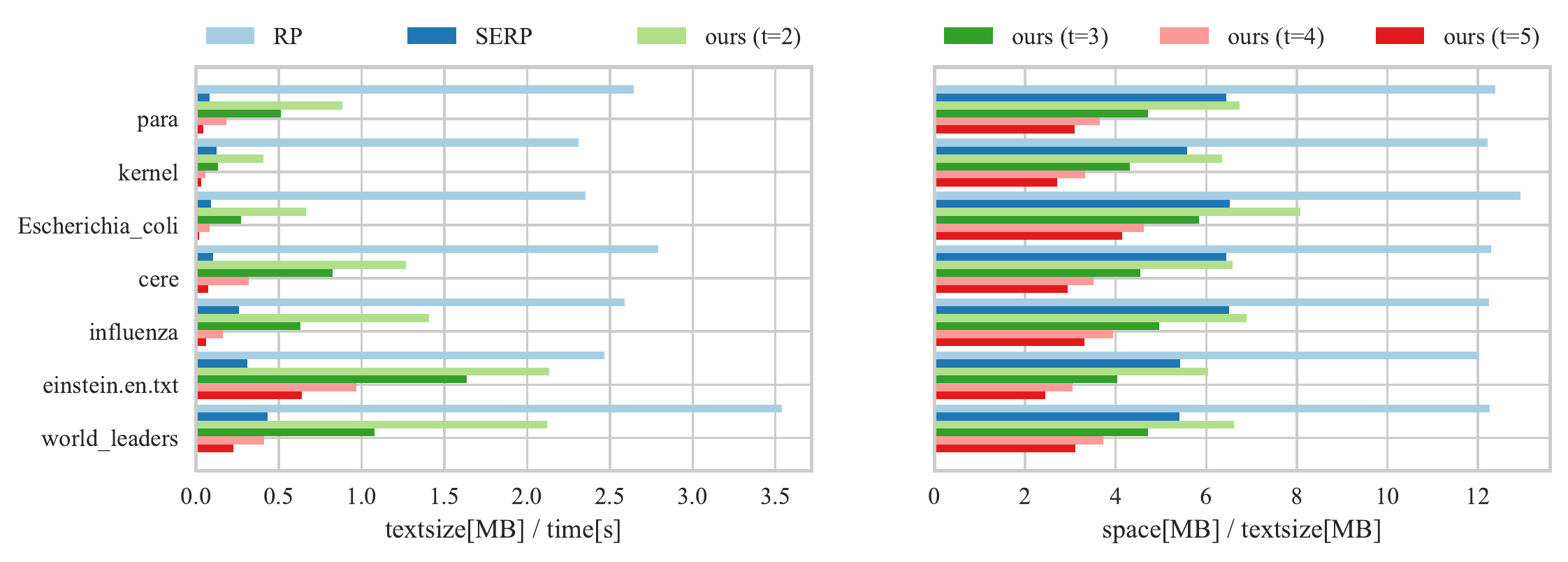}
  \end{center}
  \caption{
    Comparisons in textsize [MB] / time [s] (larger one is faster) and space [MB] / textsize [MB] (smaller one is better).
  }
  \label{fig:graph}
\end{figure}

\section{Conclusions and Future Work}\label{sec:conclusion}
We have proposed the first recompression algorithm for obtaining an output of the RePair algorithm 
via other space-saving grammar compression without decompressing it.
As a consequence, depending on the size of preliminarily compressed input text,
our recompression algorithm can simulate the RePair algorithm in compressed space.
We showed that our algorithm runs in reasonable time for several benchmarks
consisting of highly compressible texts.
Moreover, we showed that our algorithms can be used to reduce the peak memory usage of existing RePair algorithms,
and the approach outperforms the most space efficient linear-time algorithm to date.
A future work is to implement the improved version of the recompression algorithm
achieving the smaller time complexity and examine the performance of running time
compared with other implementations of RePair and its variants.
Another important future work is to prove preciser upper bound and/or lower bound
of the recompression for RePair.
An acquisition of new knowledge about the complexity
would further reduce the running time and space of the proposed algorithm.
These improvements lead us to the final goal: a faster recompression of RePair than the original one 
working in uncompressed space.

\bibliographystyle{plain}
\bibliography{ref}

\begin{thebibliography}{10}

\bibitem{MaruyamaRP}
Linear-time RePair. \url{https://code.google.com/archive/p/re-pair/downloads}.

\bibitem{SERP}
Space-efficient RePair. \url{https://github. com/nicolaprezza/Re-Pair}.

\bibitem{2018BannaiGI_OnlinLz77ParsinAndMatch_CPM}
Hideo Bannai, Travis Gagie, and Tomohiro I.
\newblock Online {LZ77} parsing and matching statistics with rlbwts.
\newblock In {\em {CPM}}, pages 7:1--7:12, 2018.

\bibitem{Bannai2013CSt}
Hideo Bannai, Pawel Gawrychowski, Shunsuke Inenaga, and Masayuki Takeda.
\newblock Converting {SLP} to {LZ78} in almost linear time.
\newblock In {\em CPM}, pages 38--49, 2013.

\bibitem{Bannai2012ELF}
Hideo Bannai, Shunsuke Inenaga, and Masayuki Takeda.
\newblock Efficient {LZ78} factorization of grammar compressed text.
\newblock In {\em SPIRE}, pages 86--98, 2012.

\bibitem{Bille2017PracticalEffectiveRePair}
Philip Bille, Inge~Li G{\o}rtz, and Nicola Prezza.
\newblock Practical and effective {Re-Pair} compression, 2017.
\newblock arXiv:1704.08558.

\bibitem{BilleGortzPrezza2017}
Philip Bille, Inge~Li G{\o}rtz, and Nicola Prezza.
\newblock Space-efficient {Re-Pair} compression.
\newblock In {\em {DCC}}, pages 171--180, 2017.

\bibitem{Charikar2005sgp}
Moses Charikar, Eric Lehman, Ding Liu, Rina Panigrahy, Manoj Prabhakaran, Amit
  Sahai, and Abhi Shelat.
\newblock The smallest grammar problem.
\newblock {\em {IEEE} Transactions on Information Theory}, 51(7):2554--2576,
  2005.

\bibitem{Cilibrasi2005ClusterByCompression}
Rudi Cilibrasi and Paul M.~B. Vit{\'{a}}nyi.
\newblock Clustering by compression.
\newblock {\em {IEEE} Trans. Information Theory}, 51(4):1523--1545, 2005.

\bibitem{Claude2007FastCompactWebGraphRepresentation}
Francisco Claude and Gonzalo Navarro.
\newblock A fast and compact web graph representation.
\newblock In {\em {SPIRE}}, pages 118--129, 2007.

\bibitem{Ganczorz2017EntropyBoundsForGrammarCompression}
Michal Ganczorz.
\newblock Entropy bounds for grammar compression, 2018.
\newblock arXiv:1804.08547.

\bibitem{Ganczorz2017ImproveOnRePair}
Michal Ganczorz and Artur Jez.
\newblock Improvements on {Re-Pair} grammar compressor.
\newblock In {\em {DCC}}, pages 181--190, 2017.

\bibitem{Gonzalez2007CompressedTextIndexesWithFastLocate}
Rodrigo Gonz{\'{a}}lez and Gonzalo Navarro.
\newblock Compressed text indexes with fast locate.
\newblock In {\em {CPM}}, pages 216--227, 2007.

\bibitem{Goto2011RCT}
Keisuke Goto, Shirou Maruyama, Shunsuke Inenaga, Hideo Bannai, Hiroshi
  Sakamoto, and Masayuki Takeda.
\newblock Restructuring compressed texts without explicit decompression, 2011.
\newblock arXiv:1107.2729.

\bibitem{Hucke2017ApproxRatioOfRePair}
Danny Hucke, Artur Jez, and Markus Lohrey.
\newblock Approximation ratio of repair, 2017.
\newblock arXiv:1703.06061.

\bibitem{I2017LceWithRecompression}
Tomohiro I.
\newblock Longest common extensions with recompression.
\newblock In {\em {CPM}}, pages 18:1--18:15, 2017.

\bibitem{Jez2012CMN}
Artur Je{\.z}.
\newblock Compressed membership for {NFA (DFA)} with compressed labels is in
  {NP (P)}.
\newblock In {\em Proc. {STACS}}, pages 136--147, 2012.

\bibitem{Jez2015Aog}
Artur Jez.
\newblock Approximation of grammar-based compression via recompression.
\newblock {\em Theor. Comput. Sci.}, 592:115--134, 2015.

\bibitem{Jez2015FFC}
Artur Je{\.z}.
\newblock Faster fully compressed pattern matching by recompression.
\newblock {\em {ACM} Transactions on Algorithms}, 11(3):20:1--20:43, 2015.

\bibitem{Jez2016OneVariableWordEquation_LinearTime}
Artur Je{\.z}.
\newblock One-variable word equations in linear time.
\newblock {\em Algorithmica}, 74(1):1--48, 2016.

\bibitem{Jez2016Recompression_WordEquations}
Artur Je{\.z}.
\newblock Recompression: {A} simple and powerful technique for word equations.
\newblock {\em J. {ACM}}, 63(1):4, 2016.

\bibitem{Jez2014Aos}
Artur Je{\.z} and Markus Lohrey.
\newblock Approximation of smallest linear tree grammar.
\newblock In {\em Proc. {STACS}}, pages 445--457, 2014.

\bibitem{Larsson1999RePair}
N.~Jesper Larsson and Alistair Moffat.
\newblock Offline dictionary-based compression.
\newblock In {\em {DCC}}, pages 296--305, 1999.

\bibitem{Lohrey2013Xts}
Markus Lohrey, Sebastian Maneth, and Roy Mennicke.
\newblock Xml tree structure compression using {RePair}.
\newblock {\em Inf. Syst.}, 38(8):1150--1167, 2013.

\bibitem{Masaki2016OnlineRePair}
Takuya Masaki and Takuya Kida.
\newblock Online grammar transformation based on {Re-Pair} algorithm.
\newblock In {\em {DCC}}, pages 349--358, 2016.

\bibitem{Navarro2011PCD}
Gonzalo Navarro, Simon~J. Puglisi, and Daniel Valenzuela.
\newblock Practical compressed document retrieval.
\newblock In {\em SEA}, pages 193--205, 2011.

\bibitem{Navarro2008RePairAchievesHighOrderEntropy}
Gonzalo Navarro and Lu{\'{\i}}s M.~S. Russo.
\newblock {Re-pair} achieves high-order entropy.
\newblock In {\em {DCC}}, page 537, 2008.

\bibitem{Nelson1995Bisection}
Greg Nelson, John Kieffer, and Pamela Cosman.
\newblock An interesting hierarchical lossless data compression algorithm.
\newblock In {\em IEEE Information Theory Soc. Workshop}, 1995.

\bibitem{2017PolicritiP_Lz77ComputBasedOnRun_Algorithmica}
Alberto Policriti and Nicola Prezza.
\newblock {LZ77} computation based on the run-length encoded bwt.
\newblock {\em Algorithmica}, Jul 2017.

\bibitem{slplz}
Wojciech Rytter.
\newblock Application of {L}empel-{Z}iv factorization to the approximation of
  grammar-based compression.
\newblock {\em Theoretical Computer Science}, 302(1-3):211--222, 2003.

\bibitem{Sekine2014AdaptiveDicShare_RePair}
Kei Sekine, Hirohito Sasakawa, Satoshi Yoshida, and Takuya Kida.
\newblock Adaptive dictionary sharing method for {Re-Pair} algorithm.
\newblock In {\em {DCC}}, page 425, 2014.

\bibitem{Tabei2016ScalablePartialLeastSquaresRegression_GrammarCompressedDataMatrices}
Yasuo Tabei, Hiroto Saigo, Yoshihiro Yamanishi, and Simon~J. Puglisi.
\newblock Scalable partial least squares regression on grammar-compressed data
  matrices.
\newblock In {\em {KDD}}, pages 1875--1884, 2016.

\bibitem{Takabatake2017Solca}
Yoshimasa Takabatake, Tomohiro I, and Hiroshi Sakamoto.
\newblock A space-optimal grammar compression.
\newblock In {\em {ESA}}, pages 67:1--67:15, 2017.

\bibitem{Wan2003BrowsingAndSearchingCompressedDocuments}
Raymond Wan.
\newblock {\em Browsing and Searching Compressed Documents}.
\newblock PhD thesis, The University of Melbourne, 2003.

\bibitem{Wan2007BlockMergingForRePair}
Raymond Wan and Alistair Moffat.
\newblock Block merging for off-line compression.
\newblock {\em {JASIST}}, 58(1):3--14, 2007.

\bibitem{Willard1983LWR}
Dan~E. Willard.
\newblock Log-logarithmic worst-case range queries are possible in space
  theta(n).
\newblock {\em Inf. Process. Lett.}, 17(2):81--84, 1983.

\bibitem{LZ77}
J.~Ziv and A.~Lempel.
\newblock A universal algorithm for sequential data compression.
\newblock {\em IEEE Transactions on Information Theory}, IT-23(3):337--349,
  1977.

\bibitem{LZ78}
J.~Ziv and A.~Lempel.
\newblock Compression of individual sequences via variable-length coding.
\newblock {\em IEEE Transactions on Information Theory}, 24(5):530--536, 1978.

\end{thebibliography}

\end{document}